\documentclass[11pt]{article}

\usepackage{fullpage}
\usepackage{float}
\usepackage{epsfig}
\usepackage{amsthm}
\usepackage{amsmath}
\usepackage{amssymb}
\usepackage{amstext}
\usepackage{xspace}
\usepackage{xcolor}
\usepackage{latexsym}
\usepackage{verbatim}
\usepackage{multirow}
\usepackage{placeins}
\usepackage{ifthen}
\usepackage{url}
\usepackage[pagebackref=true,hidelinks]{hyperref}
\usepackage{natbib}
\usepackage{graphicx}
\usepackage{times}

\usepackage[ruled,vlined]{algorithm2e}

\DeclareMathAlphabet{\mathitbf}{OML}{cmm}{b}{it}

\usepackage{savesym}
\savesymbol{pdfbookmark}
\usepackage[pagebackref=true]{hyperref}
\restoresymbol{HR}{pdfbookmark}

\usepackage{amsthm}

\newtheorem{theorem}{Theorem}[section]

\newtheorem{claim}{Claim}
\newtheorem{lemma}[theorem]{Lemma}

\newtheorem{corollary}[theorem]{Corollary}

\renewcommand{\comment}[1]{}
\newenvironment{proof}{\noindent{\em Proof:}}{\hfill\qed}

\newenvironment{proofof}[1]{\vspace{0.1in}\noindent{\em Proof of #1.}}{\hfill\qed}

\theoremstyle{definition}

\newcommand{\alloc}{x}
\newcommand{\allocs}{{\mathbf \alloc}}

\newcommand{\alloci}[1][i]{{\allocs_{#1}}}
\newcommand{\allocj}[1][j]{{\alloc_{#1}}}
\newcommand{\allocij}[1][j]{\alloc_{i#1}}

\newcommand{\val}{v}
\newcommand{\vals}{{\mathbf \val}}

\newcommand{\vali}[1][i]{{\vals_{#1}}}
\newcommand{\valj}[1][j]{{\val_{#1}}}
\newcommand{\valij}[1][j]{{\val_{i#1}}}

\newcommand{\dist}{\mathcal D}
\newcommand{\dists}{{\mathbf \dist}}
\newcommand{\distsmi}{{\mathbf \dist}_{-i}}
\newcommand{\disti}[1][i]{{\dists_{#1}}}
\newcommand{\distj}[1][j]{{\dist_{#1}}}
\newcommand{\distij}[1][j]{{\dist_{i#1}}}

\newcommand{\price}{p}
\newcommand{\prices}{{\mathbf \price}}
\newcommand{\pricei}[1][i]{{\prices_{#1}}}
\newcommand{\pricej}[1][j]{{\price_{#1}}}

\newcommand{\priceij}[1][j]{\price_{i#1}}

\newcommand{\exante}{{ex~ante}}
\newcommand{\Exante}{{Ex~ante}}
\newcommand{\ExAnte}{{Ex~Ante}}
\newcommand{\expost}{{ex~post}}

\newcommand{\eaprob}{q}
\newcommand{\eaprobs}{{\mathbf \eaprob}}
\newcommand{\eaprobi}[1][i]{{\eaprobs_{#1}}}
\newcommand{\eaprobj}[1][j]{{\eaprob_{#1}}}
\newcommand{\eaprobij}[1][j]{{\eaprob_{i#1}}}
\newcommand{\eaprobsA}[1][A]{\eaprobs^{#1}}
\newcommand{\eaprobAj}[1][j]{\eaprob^A_{#1}}

\newcommand{\eaprice}{\beta}
\newcommand{\eaprices}{{\mathbf \eaprice}}

\newcommand{\eapricej}[1][j]{{\eaprice_{#1}}}
\newcommand{\eapriceij}[1][j]{{\eaprice_{i#1}}}

\newcommand{\feas}{{\mathcal F}}
\newcommand{\feasi}[1][i]{\feas_{#1}}
\newcommand{\udfeas}{\feasi[\mathsc{UnitDemand}]}
\newcommand{\adfeas}{\feasi[\mathsc{Additive}]}
\newcommand{\ptope}{{\mathcal P}_{\feas}}
\newcommand{\ptopei}[1][i]{{\mathcal P}_{\feas_{#1}}}

\newcommand{\I}{\mathcal{I}}

\newcommand{\prob}[2][]{\text{\bf Pr}\ifthenelse{\not\equal{}{#1}}{_{#1}}{}\!\left[#2\right]}
\newcommand{\expect}[2][]{\text{\bf E}\ifthenelse{\not\equal{}{#1}}{_{#1}}{}\!\left[#2\right]}

\newcommand{\var}{Var}

\newcommand{\mech}{{\mathcal M}}
\newcommand{\ef}{\pi}
\newcommand{\efi}[1][i]{\pi_{#1}}

\newcommand{\rev}{\mathsc{Rev}}
\newcommand{\revm}[1][\mech]{\rev^{#1}}
\newcommand{\revt}[1][(\ef,\prices)]{\rev^{#1}}
\newcommand{\brev}{\mathsc{BRev}}
\newcommand{\trev}{\mathsc{TRev}}
\newcommand{\srev}{\mathsc{SRev}}
\newcommand{\earev}[1][\eaprobs]{\rev_{#1}}

\newcommand{\easrev}[1][\eaprobs]{\srev_{#1}}
\newcommand{\eatrev}[1][\eaprobs]{\trev_{#1}}

\newcommand{\Val}{\mathsc{Val}}
\newcommand{\Valm}[1][\mech]{\Val^{#1}}
\newcommand{\eaVal}[1][\eaprobs]{\Val_{#1}}

\newcommand{\dy}{\,dy}

\newcommand{\union}{\cup}
\newcommand{\intersect}{\cap}

\newcommand{\mathsc}[1]{{\normalfont\textsc{#1}}}

\newcommand{\x}{\mathbf{x}}

\newcommand{\fj}{f_j}

\newcommand{\threshj}{t_j}
\newcommand{\tconst}{\tau}
\newcommand{\core}{\dists^C_\emptyset} 
\newcommand{\coreA}{\dists^C_A}
\newcommand{\corej}{\dist^C_j}
\newcommand{\tail}{\dists^T_A}
\newcommand{\tailj}{\dist^T_j}

\newcommand{\tprobA}[1][A]{\rho_{#1}}
\newcommand{\tprobz}{\tprobA[\emptyset]}
\newcommand{\tprobj}[1][j]{\xi_{#1}}
\newcommand{\tprobs}{{\mathbf \xi}}

\newcommand{\tvalsmi}{\tilde{\mathbf \val}_{-i}}
\newcommand{\half}{\frac 12}
\newcommand{\efs}{{\mathbf \ef}}

\title{Mechanism Design for Subadditive Agents via an Ex-Ante Relaxation}
\author{Shuchi Chawla\thanks{University of Wisconsin--Madison:
\texttt{shuchi@cs.wisc.edu}} \and J. Benjamin Miller\thanks{University of
Wisconsin--Madison: \texttt{bmiller@cs.wisc.edu}}}

\begin{document}
                       
\maketitle

\begin{abstract}

  We consider the problem of maximizing revenue for a monopolist
  offering multiple items to multiple heterogeneous buyers. We develop
  a simple mechanism that obtains a constant factor approximation
  under the assumption that the buyers' values are additive subject to
  a matroid feasibility constraint and independent across items. Importantly,
  different buyers in our setting can have different constraints on
  the sets of items they desire. Our mechanism is a sequential variant
  of two-part tariffs. Prior to our work, simple approximation
  mechanisms for such multi-buyer problems were known only for the
  special cases of all unit-demand or all additive value buyers.

  Our work expands upon and unifies long lines of work on unit-demand
  settings and additive settings. We employ the \exante\ relaxation
  approach developed by~\citet{Alaei-FOCS11} for reducing a
  multiple-buyer mechanism design problem with an \expost\ supply
  constraint into single-buyer problems with \exante\ supply
  constraints. Solving the single-agent problems requires us to
  significantly extend a decomposition technique developed in the context of additive
  values by \citet{LY-PNAS13} and its extension to subadditive
  values by \citet{rw-15}.



\end{abstract}

\newpage

\section{Introduction}



Multi-parameter optimal mechanism design is challenging from both a
computational and a conceptual viewpoint, even when it involves only a single
buyer. Multi-parameter type spaces can be exponentially large, and
multi-dimensional incentive constraints lack the nice structure of
single-dimensional constraints that permits simplification of the optimization
problem. As a result, optimal mechanisms can possess undesirable properties such
as requiring randomness \citep{BCKW-SODA10, HN-EC12, HN-EC13}, displaying
non-monotonicity of revenue in values \citep{HR-2013, rw-15}, and are in many
cases computationally hard to find (see, e.g., \citealp{DDT-WINE12,
DDT-SODA14}).  The situation exacerbates in multi-agent settings.  \citet[chap.
8]{Hartline-MDnA} identifies two further difficulties: multi-parameter agents
impose multi-dimensional externality on each other that may not be possible to
capture succinctly; and multi-parameter problems are typically not revenue
linear, meaning that the optimal revenue does not scale linearly with the
probability of service. Designing simple near-optimal mechanisms in such
settings is a primary goal of algorithmic mechanism design.

In this paper we study the problem facing a monopolist with many items
and many buyers, where each buyer is interested in buying one of many
different subsets of items, and his value for each such subset is
additive over the items in that subset. What selling mechanism should
the monopolist use in such a setting to maximize his revenue?  One
challenge for the seller is that buyers may have heterogeneous
preferences: some buyers are interested in buying a few specific
items, others are indifferent between multiple items, and yet others
have a long shopping list. We design the first approximation mechanism
for this problem; our main result is a constant-factor approximation when
buyers' values are additive up to a matroid feasibility
constraint.

Our approximation mechanism has a particularly simple and appealing
format -- a sequential extension of standard {\em two-part tariff
  mechanisms}.  Two-part tariffs for a single agent have the following
structure. The buyer first pays a fixed entry fee and is then allowed
to purchase any set of items at fixed per-item prices. The buyer may
choose not to participate in the mechanism, in which case he does not
pay the entry fee and does not receive any item. In our context,
buyers pay (different) entry fees at the beginning, and then take
turns (in an arbitrary but fixed order) to buy a subset of items at
predetermined item-specific prices, subject to availability. There are
many real-world examples of two-part tariffs, such as amusement park
pricing; memberships for discount shopping clubs like Costco, Sam's
Club, and Amazon's Prime; telephone services; and membership programs
for cooperatives and CSAs. These mechanisms have long been studied in
economics for their ability to effectively price discriminate among
different buyers despite their relative simplicity. \citet{Arm-99}
shows, for example, that for an additive value buyer with independent
item values and sufficiently many items, two-part tariffs extract
nearly the entire social surplus.

Our work combines and significantly extends techniques from several
different lines of work in mechanism design. We use the {\em \exante\
  relaxation} of \citet{Alaei-FOCS11} to break up the multi-agent
revenue maximization problem into its single-agent counterparts and
capture the externalities among buyers through \exante\ supply
constraints.  We solve the single-agent problems with \exante\ supply
constraints by adapting and extending the so-called {\em core-tail
  decomposition} technique of \citet{LY-PNAS13}, as well employing the
prophet inequalities of \citet{CHMS-STOC10} and
\citet{fsz-15}. Finally, we use ideas from \citep{CHMS-STOC10} to
combine the resuting single-agent mechanisms sequentially and obtain a
multi-agent approximation mechanism that is \expost\ supply~feasible.
While our main result applies to buyers with values additive up to a
matroid constraint, parts of our approach extend to more general value
functions such as those satisfying the gross substitutes condition.

\subsection{Multi-Parameter Mechanism Design: Previous Work}

This paper belongs to a long line of research on finding simple and
approximately optimal mechanisms for multi-parameter settings under
various assumptions on the buyers' value functions and type
distributions, and on the seller's supply constraint.
The first breakthrough along these
lines was made by \citet{CHK-07} who showed that the revenue of an
optimal mechanism for a single unit-demand buyer can be approximated
within a factor of $3$ by an item pricing,\footnote{\citet{CHMS-STOC10}
  later improved this approximation factor to $2$.} a mechanism that
allows the buyer to choose any item to buy at fixed per-item
prices. More recently, \citet{bilw-focs14} developed a similar result
for a single buyer with additive values.\footnote{This is the
  culmination of a series of papers including \cite{HN-EC12, HN-EC13,
    LY-PNAS13}.}  They showed that the revenue of an optimal mechanism
in this case is approximated within a factor of $6$ by one of two
simple mechanisms: an item pricing that fixes a price for each item
and allows the buyer to choose any subset of items to buy, and a
bundle pricing that allows the buyer to buy the grand bundle of all
items at a fixed price. Observe that item pricing and bundle pricing
are both two-part tariffs (with the entry fee or the per-item prices
being zero, respectively).

Unit-demand and additive types are two extremes within a broader class
of value functions that we call {\em constrained additive values}. A
constrained additive buyer has a value (drawn independently) for each
item under sale; he is interested in buying a set of items that
satisfies a certain downward-closed constraint; his value is additive
over any such set. We have only recently begun to understand optimal
mechanism design for a single agent with constrained additive
values. \citet{rw-15} proved that in this setting, as in the additive
case, either item pricing or bundle pricing gives a constant-factor
approximation to the optimal revenue.\footnote{\citet{rw-15}'s result
  holds for a much broader setting with a single subadditive value
  agent, but their factor of approximation is rather large -- about
  340.} There are many similarities between the two lines of work on
unit-demand buyers and additive buyers, and \citeauthor{rw-15}'s
result can be seen as a unification of the two approaches, albeit with
a worse approximation factor.


Multi-parameter settings with multiple buyers are less well
understood. For settings with many unit-demand buyers,
\citet{CHMS-STOC10, cms-10} developed a generic approach for
approximation via sequential posted-price mechanisms (SPMs). SPMs
approach buyers in some predetermined order and offer items for sale
to each buyer at predetermined prices while supplies last. For
settings with many additive-value buyers, \citet{yao-15} showed that
either running a second-price auction for each item separately or
optimally selling to bidder $i$ the set of items for which he is the
highest bidder\footnote{In the latter case, Yao approximates the
  optimal revenue via two-part tariffs.}
achieves a constant-factor approximation. \citet{CDW-16} presented a
new uniform framework that can be used to rederive both Yao and Chawla
et al.'s results, with a tighter analysis for the former. However,
prior to our work, no approximations were known for other constrained
additive settings or for settings with heterogeneous buyers. Consider,
for example, a setting with some unit-demand and some additive
buyers. In this case, neither of the results mentioned above provide
an approximation.  \citeauthor{CHMS-STOC10}'s analysis relies on a
reduction from multi-dimensional incentive constraints to
single-dimensional ones that applies only to the unit-demand setting,
and, in particular, cannot account for revenue from bundling, which is
crucial in non-unit-demand settings. \citeauthor{yao-15}'s approach on
the other hand relies on allocating each item to the highest value
agent, and cannot provide a constant-factor approximation for
subadditive~agents.\footnote{To see why Yao's approach cannot work for
  unit-demand agents, observe that if a single unit-demand agent has
  the highest value for each item, the seller must try to sell all but
  one item to non-highest-value buyers in order to obtain good
  revenue.}



A different approach to optimal mechanism design due to
\citet{CDW-STOC12, CDW-FOCS12, CDW-SODA13, CDW-FOCS13} uses
linear programming formulations for settings with small support type
distributions, and shows that optimal mechanisms are virtual welfare
maximizers. This approach is unsuitable for our setting which, even
upon discretization of values, involves distributions over exponential
size supports. Moreover, mechanisms generated by this approach tend to
lack the nice structure and simplicity of pricing-based~mechanisms.

Finally, a new approach to mechanism design has emerged in recent
years that uses duality theory to design as well as analyze optimal or
approximately optimal mechanisms \citep[see, e.g.,][]{DDT-EC15, GK-14, GK-15, HH-15,
  CDW-16}. Designing good dual solutions in this context, however,
involves more art than science, and for the most part, positive
results are restricted to very special classes of value functions and
value distributions.


\subsection{Our Techniques and Contributions}
\paragraph{\ExAnte\ Relaxation.}
Our work follows a generic approach developed by \citet{Alaei-FOCS11}
for transforming multi-agent mechanism design problems into their
single-agent counterparts via the so-called {\em \exante\ relaxation}.
In a multi-agent setting, agents impose externalities upon each other
through the seller's supply constraint: each item must be sold to at
most one buyer \expost. Alaei proposes relaxing the problem by
enforcing the supply constraint \exante\ rather than \expost: the
probabilities with which an item is sold to the different buyers
should sum up to no more than one. In other words, in expectation the
item is sold at most once. Applying the \exante\ relaxation to a
mechanism design problem with multiple buyers involves three steps:
\begin{enumerate}
  \item {\bf Decompose into single-agent problems:} determine the \exante\
  probabilities with which each item can be sold to each buyer; for each item
  these probabilities should sum up to no more than $1$;

  \item {\bf Solve single-agent problems:} for each agent, find an approximately
  optimal mechanism satisfying the \exante\ supply constraint determined in the
  first step;

  \item {\bf Stitch single-agent mechanisms:} combine the single-agent
  mechanisms developed in the second step into a single mechanism that satisfies
  the supply constraints \expost.
\end{enumerate}
The first step is conceptually simple and applies in any setting where
buyers have independent values. We reproduce this argument in
Section~\ref{sec:relaxation} for completeness.




Alaei described how to implement the second and third steps for
problems involving unit-demand agents.\footnote{Alaei also presented
  solutions for certain additive-value settings under the assumption
  that the agents' type spaces are small and given explicitly.} For
the third ``stitching'' step, he suggested composing the
single-agent mechanisms sequentially (similar to the approach of
\citet{CHMS-STOC10}).  However, this does not work for arbitrary
single-agent mechanisms. Once the composite mechanism has sold off a
few items, fewer bundles are available to subsequent buyers, and the
mechanism may obtain far less revenue than its single-agent
counterparts. We show that two-part tariffs compose well without
much loss in revenue when each buyer's value function is additive
up to a matroid feasibility constraint (and, more generally, when the
value functions satisfy the gross substitutes condition).

\paragraph{Core-Tail Decomposition.}
In order to bound the single-agent revenue as required in step two of
the \exante\ approach, we use the core-tail decomposition of
\citet{LY-PNAS13}, and its extensions due to \citet{bilw-focs14} and
\citet{rw-15}. Roughly speaking, in the absence of \exante\ supply
constraints, for any vector of item values, we can partition items
into those with small value and those with large value.  This
partitioning is done in such a manner that the set of large-value
items (a.k.a. the tail) contains only a few items in expectation; the
revenue generated by these items behaves essentially like unit-demand
revenue, and can be recovered by selling the items separately via an
argument of \citet{cms-10}.  The set of small-value items (a.k.a. the
core), on the other hand, displays concentration of value and the
revenue generated by these items can be recovered via bundling
\citep{rw-15}.

Under an \exante\ supply constraint the revenue generated by the tail
can still be recovered via item pricing as before. Bounding the
revenue from the core is trickier, however, because different items
may face very different \exante\ constraints, and their total values
may not concentrate well. Furthermore, selling the grand bundle
allocates all items with the same probability to the buyer and
consequently may not respect the given \exante\ constraint. We make a
careful choice of thresholds for partitioning values into the core and
the tail in such a manner that we can recover the value of the core in
two parts:
(1) when the \exante\ constraint is strong (i.e. the allocation
probabilities are mostly small), selling separately recovers most of
the core revenue; (2) when the \exante\ constraint is weak (i.e. the
allocation probabilities are mostly large), bundling as part of a
two-part tariff recovers most of the core revenue while continuing to
respect the \exante\ constraint.

\paragraph{Prophet Inequalities.}
Observe that the \exante\ approach described above relaxes the
seller's supply constraint, but continues to enforce the buyer's
demand constraint\footnote{The buyer's demand constraint refers to,
  e.g., whether the buyer desires one item as in the unit-demand case,
  or all items as in the additive case.} \expost. It is unclear how a
relaxation of the buyer's demand constraint would capture revenue due
to bundling, and whether such a relaxation is useful for mechanism
design. Nevertheless, our analysis gives rise to a term which
corresponds to item-pricing revenue from a common relaxation of the
seller's and buyer's constraints.
Roughly speaking, this term captures the total revenue that the seller
can obtain from the buyer by selling each item separately subject to a
bound on the probability of sale, under the condition that these
bounds respect both the seller's and the buyer's feasibility
constraints in an \exante\ sense. For example, for a unit-demand
buyer, the probabilities of sale over the items must sum up to no more
than $1$. We then employ a prophet inequality to relate this term to
the optimal item-pricing revenue for that buyer. A prophet inequality
in this context specifies an item pricing that, regardless of which
maximal feasible set of items the buyer purchases, obtains in
expectation a constant fraction of the \exante\ optimal
revenue. Prophet inequalities of the above form are known to hold for
several classes of feasibility constraints, such as uniform matroids,
partition matroids, and their intersections (see, e.g.,
\citealp{CHMS-STOC10}). For general matroid constraints, it is not
known whether a prophet inequality with static item prices as
described above can obtain a constant approximation
factor.\footnote{\citet{KW-STOC12} present a prophet inequality with
  adaptive prices, but this is unsuitable for our
  setting.} However, \citet{fsz-15} give a prophet inequality that obtains a
constant approximation by restricting the buyer's demand -- in other words, by
forbidding the buyer to purchase certain feasible sets. We discuss and use these
results in Section~\ref{sec:single-agent}.

\paragraph{The Final Mechanism.}
As mentioned earlier, our final mechanism is a sequential two-part
tariff mechanism. We remark that buyers in our mechanism are required
to pay the entry fee before finding out whether their favorite items
will be available when it is their turn to buy; therefore, our
mechanism is only Bayesian incentive compatible (BIC), and not
necessarily dominant strategy incentive compatible (DSIC). We leave
open the question of whether it is possible to approximate the optimal
revenue within a constant factor via a DSIC mechanism. In some
settings, our mechanism restricts the subsets of items that a buyer is
allowed to buy; we call such a mechanism a {\em demand-limiting
  sequential two-part tariff}. This is seen, for instance,
in market-style CSA programs in which members can buy only certain
quantities and combinations of~produce.

\paragraph{Other Contributions.} As special cases of our general
result, we also obtain improvements to the results of \citet{rw-15}.
Recall that \citeauthor{rw-15} show that for a single buyer with
subadditive values, either item pricing or bundle pricing obtains a
constant-factor approximation. We improve this result in two
ways. First, for constrained additive values, we improve the
approximation factor from about 340 to 31.1
(Corollary~\ref{cor:true-single-agent}).\footnote{It is possible to
  use \citeauthor{rw-15}'s techniques to obtain a better approximation
  for the special case of constrained additive values, however, the
  resulting bound is still much weaker than ours.}  Second, we show
that the result holds also under an \exante\ constraint for a suitable
definition of item pricings and bundle pricings that respect the same
\exante\ constraint (see
Corollary~\ref{cor:general-ex-ante}). Finally, for revenue
maximization with multiple additive buyers, we adapt arguments from
\citep{bilw-focs14} to obtain an approximation factor of 28
(Appendix~\ref{sec:additive}); 
this is an improvement over \citet{yao-15}'s
approximation factor of 69 for the same setting, but is worse than
\citet{CDW-16}'s improvement of \citeauthor{yao-15}'s analysis to an
8-approximation. Arguably, our analysis for this setting is
conceptually simpler than both of those~works.

\paragraph{Symmetric Settings.} In an interesting special case of our
setting, the buyers are a~priori symmetric (but items are
heterogeneous). That is, each buyer has a value vector drawn from
identical independent distributions, and also desires the same bundles
of items. In this setting, our mechanism sets the same entry fee as
well as item prices for all buyers. Furthermore, these fees and prices
can be computed~efficiently (Section~\ref{sec:symmetric}).

\paragraph{Further Directions.} For settings with asymmetric buyers, we leave
open the question of efficiently solving the \exante\ relaxation.
Our main result requires buyers' demand constraint to be matroids for two
reasons: this allows us to use a prophet inequality for a single agent, and it
also enables us to combine single-agent mechanisms sequentially without much
loss in revenue.  It is an interesting challenge to apply the \exante\ approach
for demand constraints beyond matroids, or for more general classes of
subadditive values.


\section{Preliminaries}
\label{sec:prelim}

We consider a setting with a single seller and $n$ buyers. The seller has $m$
heterogeneous items to sell. Each buyer $i\in [n]$ has a type composed of a
public downward-closed demand constraint $\feasi\subseteq 2^{[m]}$ and a private
value vector $\vali=(\valij[1], \cdots, \valij[m])$ that maps items to
non-negative values.  Roughly speaking, the demand constraint $\feasi$ describes
the maximal sets of items from which the buyer derives value. Formally, the
buyer's value for a set of items is described by a {\em constrained additive}
function: for $S\subseteq 2^{[m]}$,
\[ \vali(S) = \max_{S'\in\feasi; S'\subseteq S} \sum_{j\in S'} \valij\]

It will sometimes be necessary to consider feasibility restricted to subsets of
the available items. For $M' \subseteq [m]$, the {\em restriction of $\feasi$ to
$M'$}, denoted $\feasi|_{M'}$, is formed by dropping items not in $M'$.
Formally, $\feasi|_{M'} = \feasi\intersect2^{M'}$.  We will typically assume
that for all $i$, $\feasi$ is a matroid; see Appendix~\ref{sec:matroids} for a
review of matroid concepts.

We assume that the values $\valij$ are drawn from distribution $\distij$
independently of all other values; we use $\disti=\prod_j \distij$ to denote the
joint distribution of buyer $i$'s value vector and $\dists = \prod_i \disti$ to
denote the joint distribution over all value vectors. The demand constraints
$\feasi$ may be different for different buyers. Let $\feas =
\{\feasi\}_{i\in[n]}$ denote the tuple of feasibility constraints, one for each
buyer. 

\subsection{Incentive Compatible Mechanisms and Revenue Maximization}

A mechanism $\mech$ takes as input the value vectors $\vals=(\vali[1], \cdots,
\vali[n])$ and returns an allocation $\allocs(\vals)$ and payment vector
$\prices(\vals)$. Here $\alloci(\vals)$ denotes the (potentially random) set of
items that is allocated to buyer $i$. A mechanism $\mech$ is {\em
supply-feasible} if every item is allocated to at most one buyer; in other
words, for all $\vals$, and $i_1\ne i_2$,
$\alloci[i_1](\vals)\cap\alloci[i_2](\vals) = \emptyset$ with probability $1$.

We use $\allocij(\vals)$ to denote the probability with which buyer $i$ receives
item $j$.  Without loss of generality, we focus on mechanisms that for every
value vector $\vals$ and every buyer $i$ satisfy $\alloci(\vals)\in\feasi$ with
probability $1$; we call such mechanisms {\em demand-feasible}.  Consequently,
we note that the vector $(\allocij[1](\vals), \cdots, \allocij[m](\vals))$ lies
in the polytope enclosing $\feasi$,  which we denote\footnote{Formally,
  $\ptopei$ is the convex hull of the incidence vectors of all sets in $\feasi$
  in $\Re^m$.} $\ptopei$.
In the rest of the paper we will overload notation and use
$\alloci(\vals)$ to denote the vector $(\allocij[1](\vals), \cdots,
\allocij[m](\vals))$.

We assume that buyers are risk neutral and have quasi-linear utilities. In other
words, the utility that a buyer derives from allocation $\alloci$ and payment
$\pricei$ is given by $\alloci\cdot\vali - \pricei$. We consider mechanisms
which are {\em Bayesian incentive compatible (BIC)}. A mechanism is BIC if
truthtelling is a Bayes-Nash equilibrium; that is, if a buyer maximizes his own
utility---in expectation over other buyers' values, assuming they report
truthfully, as well as randomness inherent in the mechanism---by reporting
truthfully. In contrast, a mechanism is {\em dominant-strategy incentive
compatible (DSIC)} if truthtelling is a dominant strategy; that is, if a buyer
maximizes his own utility by reporting truthfully, regardless of what other
buyers report.

We are interested in revenue maximization for the seller. The seller's revenue
from a BIC mechanism $\mech=(\allocs,\prices)$ at value vectors $\vals$ is
$\sum_i\pricei(\vals)$, and the expected revenue is $\revm(\dists) =
\expect[\vals\sim\dists]{\sum_i\pricei(\vals)}$. The revenue maximization
problem seeks to maximize $\revm(\dists)$ over all BIC mechanisms that are
demand- and supply-feasible; we use $\rev(\dists,\feas)$ to denote this maximum
revenue.


\subsection{\ExAnte\ Constrained Revenue Maximization}
We will reduce the multiple buyer revenue maximization problem
described above to single-buyer problems with \exante\ supply
constraints. The following definitions are for a single agent $i$; we
omit the subscript $i$ for clarity.  Let $\eaprobs = (\eaprobj[1],
\cdots, \eaprobj[n])$ be a vector of probabilities with $\eaprobj\in
[0,1]$ for all $j\in [m]$. A mechanism $\mech=(\allocs,\prices)$ is
{\em $\eaprobs$-constrained under $\dists$} if for all items $j\in
[m]$, its \exante\ probability for selling item $j$ when values are
drawn from $\dists$, $\expect[\vals\sim\dists]{\allocj(\vals)}$, is at
most $\eaprobj$.  We will consider both revenue and welfare
maximization problems over $\eaprobs$-constrained mechanisms.
Formally, we define 
\begin{align}
  \label{eq:constrained-revenue}
\earev(\dist,\feas) & = \max_{\mech=(\allocs,\prices):
\expect[\vals\sim\dists]{\allocj(\vals)}\le\eaprobj \,\,\forall j\in [m]}
\revm(\dists)
\end{align}
and
\begin{align*}
\eaVal(\dist,\feas) & = \max_{\mech=(\allocs,\prices):
\expect[\vals\sim\dists]{\allocj(\vals)}\le\eaprobj \,\,\forall j\in [m]}
\Valm(\dists),
\end{align*}
where the maximum is taken over all BIC demand-feasible mechanisms\footnote{We
don't need to impose the supply-feasibility constraint explicitly --- this is
already implicit in the \exante\ probability constraint.} and
$\Valm(\dists) = \expect[\vals\sim\dists]{\allocs(\vals)\cdot\vals}$.

It will sometimes be convenient to express the \exante\ constraint in the form
of \exante\ prices defined as: $\eapricej = \distj^{-1}(1-\eaprobj)$.  In other
words, for every $j\in [m]$, $\eapricej$ is defined such that the probability
that $\valj$ exceeds this price is precisely $\eaprobj$. Note that there is a
one-one correspondence between \exante\ probabilities and \exante\ prices.

\subsection{Special Single-Agent Mechanisms}

\paragraph{Item Pricing.} An item pricing is defined by a set of
prices $\pricej$, one for each item $j$.  A buyer is allowed to select
as many items as he pleases, up to some downward-closed constraint
$\feas$, and he pays the sum of the associated prices.  That is, if the
buyer selects the set $S \subseteq [m]$, he pays $\sum_{j\in S}p_j$.
The buyer then selects the set $S \in \feas$ which maximizes
$\sum_{j\in S}(\valj - \pricej)$. We use $\srev(\dists,\feas)$ to
denote the optimal revenue obtainable by any item pricing from a buyer
with value distribution $\dists$ and demand constraint $\feas$.



\paragraph{Bundle Pricing.} A bundle pricing is defined by a single
price (a.k.a. entry fee) $\ef$. A buyer can buy any subset of items
satisfying the demand constraint $\feas$ at price $\ef$.  A rational
buyer chooses to participate (i.e. pay the fee) if
$\val([m])=\max_{S\in\feas}\val(S) \geq \ef$ and then selects a
corresponding maximal set $S$. We use $\brev(\dists,\feas)$ to
represent the optimal revenue obtainable by any bundle pricing from a
buyer with value distribution $\dists$ and demand constraint $\feas$.


\paragraph{Two-Part Tariffs.} A two-part tariff is a common
generalization of both item pricings and bundle pricings. It is
described by an $m+1$ dimensional vector of prices: $(\ef, \pricej[1],
\cdots, \pricej[m])$. The mechanism offers each set $S\subseteq [m]$
of items to the buyer at a price of $\ef+\sum_{j\in S} \pricej$; the
buyer can then choose to buy his favorite set at these offered
prices. Informally speaking, the mechanism charges the buyer an {\em
  entry fee} of $\ef$ for the right to buy any set of items, with item
$j$ offered at a fixed price of $\pricej$. Like other pricing-based
mechanisms, two-part tariffs are deterministic, dominant strategy
incentive compatible mechanisms.
 
A utility-maximizing buyer with values $\vals$ and feasibility constraint
$\feas$ when offered a two-part tariff $(\ef, \prices)$ buys the set $S\in\feas$
of items that maximizes $\vals(S)-\ef-\sum_{j\in S} \pricej$, if  that quantity
is non-negative\footnote{This is essentially an ex-post IR condition.}; in that case, we say that the buyer participates  in the
mechanism. We denote the revenue of a two-part tariff $(\ef, \prices)$  offered
to a buyer with feasibility constraint $\feas$ and value distribution $\dists$
by $\revt(\dists,\feas)$.  We use $\trev(\dists,\feas)$ to denote the  optimal
revenue that a two-part tariff can obtain from a buyer with value  distribution
$\dists$ and demand constraint $\feas$.

Two-part tariffs are known to be approximately optimal in certain
single-agent settings. The following results\footnote{Here $\udfeas =
 \{S\subset [m] \mid |S|=1\}$ represents a unit-demand buyer, and
 $\adfeas = 2^{[m]}$ represents a buyer with fully additive values.}
are due to \citet{cms-10} and \citet{bilw-focs14}
respectively. \citet{rw-15} proved a similar result for constrained
additive values, but with a very large approximation factor (about
340).
\begin{align*}
 \rev(\dists,\udfeas) & \leq 4\,\trev(\dists,\udfeas)\\
 \rev(\dists,\adfeas) & \leq 6\,\trev(\dists,\adfeas)
\end{align*}

%

\paragraph{Pricings with an \ExAnte\ Constraint.}
Next we extend the above definitions to respect \exante\ supply
constraints. We say that a two-part tariff $(\ef, \prices)$ satisfies
\exante\ constraint $\eaprobs$ if for all $j$,
$\pricej\ge \eapricej=\distj^{-1}(1-\eaprobj)$. Note that this is a
stronger condition than merely requiring that the mechanism allocates
item $j$ with \exante\ probability at most $\eaprobj$. We use
$\eatrev(\dists,\feas)$ to denote the optimal revenue achieved by a
demand-feasible two-part tariff that satisfies \exante\ constraint
$\eaprobs$. Likewise, we use $\easrev(\dists,\feas)$ to denote the
optimal revenue achievable by an item pricing $\prices$ with
$\pricej\ge \eapricej$ for all $j$.


\subsection{Multi-Agent (Sequential) Two-Part Tariff Mechanisms} We
now extend the definition of two-part tariffs to multi-agent settings.
Consider a setting with $n$ agents and demand constraints
$\feas=\{\feasi\}_{i\in[n]}$. A {\em sequential two-part tariff} for
this setting is parameterized by an ordering $\sigma$ over the agents,
a set of entry fees $\efs=(\efi[1], \cdots, \efi[n])$, and a set of
prices $\prices=\{\priceij\}$.  The mechanism proceeds as follows.
\begin{enumerate}
  \item The ordering $\sigma$ and prices $\efs;\prices$ are announced.

  \item Each agent $i$ independently decides whether or not to participate in
  the mechanism. If the agent decides to participate, then he pays his
  corresponding entry fee~$\efi$.

  \item The mechanism considers agents in the order given by $\sigma$. When an
  agent $i$ is considered, if the agent previously declined to participate, no
  items are allocated and no payment is charged. Otherwise, of the items
  unallocated so far, the agent is allowed to purchase his favorite feasible set
  of items at the prices $\priceij$.
\end{enumerate}
Observe that agents choose whether or not to participate in the
mechanism before knowing which items will be available when it is
their turn to purchase. Accordingly, a sequential two-part tariff is
BIC but not necessarily DSIC.

The sequential two-part tariff mechanisms that we develop in this
paper are {\em order oblivious} in the sense that their revenue
guarantees hold regardless of the ordering $\sigma$ chosen over the
agents. Accordingly, in describing these mechanisms, we need only
specify the prices $\efs;\prices$.

In some cases, our two-part tariff mechanisms disallow agents from
buying certain sets of items. Specifically, a {\em demand-limiting
  sequential two-part tariff} is parameterized by an ordering
$\sigma$, prices $\efs; \prices$, as well as feasibility constraints
$\feas'=\{\feasi'\}_{i\in[n]}$ where, for every agent $i$,
$\feasi'\subseteq \feas$ is a matroid constraint stronger than the
agent's original demand constraint. When it is agent $i$'s turn to buy
items, the agent is allowed to buy any subset of items in
$\feasi'$. In particular, the agent is not allowed to buy sets of
items in~$\feasi\setminus\feasi'$.

\section{Main Results}
\label{sec:theorems}

We now state our three main results corresponding to the three parts
of the \exante\ approach for approximating $\rev(\dists, \feas)$.
Lemma~\ref{lem:relaxation} corresponds to the first {\bf relaxation}
step, and states that the revenue $\rev(\dists, \feas)$ can be bounded
by the sum of single-agent revenues with appropriate \exante\
constraints. While the lemma is stated here for buyers with
constrained additive values, it holds for arbitrary value functions as
long as values are independent across buyers.

\begin{lemma}[\bf Relaxation]
\label{lem:relaxation}
 For any feasibility constraints $\feas=\{\feasi\}_{i\in[n]}$ and value
 distributions $\dists=\prod_i\disti$, there exist \exante\ probability vectors
 $\eaprobi[1], \cdots, \eaprobi[n]$, satisfying: (1)
 $\eaprobi\in\ptopei$ for all $i$, and, (2) $\sum_i \eaprobij\le 1$
 for all $j$, such that 
 \[\rev(\dists,\feas)\le \sum_i \earev[\eaprobi](\disti,\feasi).\]
\end{lemma}

Lemma~\ref{lem:stitching-trevs} corresponds to the last {\bf stitching} step, and
shows that any single-agent two-part tariff mechanisms that
collectively satisfy an \exante\ constraint on every item can be
stitched together into a multi-agent sequential two-part tariff
mechanism without losing much revenue.

\begin{lemma}
\label{lem:stitching-trevs} 
For every agent $i$, let $\mech_i = (\efi,\pricei)$ be any two-part
tariff that is demand-feasible with respect to a matroid feasibility
constraint $\feasi$ and that satisfies \exante\ supply constraints
$\eaprobi$ under value distribution $\disti$. Let
$\feas=\{\feasi\}_{i\in[n]}$ and $\dist=\prod_i\disti$.  Then, if
$\sum_i \eaprobij\le 1/2$ for all $j$, there exists a sequential
two-part tariff mechanism $\mech$ that is supply-feasible and
demand-feasible with respect to $\feas$ such that
\[\revm(\dists)\ge \half\sum_i \revm[\mech_i](\disti).\]
\end{lemma}
 
We therefore obtain the following corollary.

\begin{corollary}[\bf Stitching]
\label{cor:stitching}
For any value distributions $\dists=\prod_i\disti$ and feasibility constraints
$\feas=\{\feasi\}_{i\in[n]}$, where each $\feasi$ is a matroid, let
$\eaprobi[1], \cdots, \eaprobi[n]$ be any \exante\ probability vectors
satisfying $\sum_i \eaprobij\le 1/2$ for all $j$. Then, there exists a
demand- and supply-feasible sequential two-part tariff mechanism
$\mech$ such that 
\[\revm(\dists)\ge \half\sum_i \eatrev[\eaprobi](\disti,\feasi).\]
\end{corollary}

In order to put together the Relaxation Lemma and the Stitching
Corollary, it remains to relate $\earev$ for a single agent to
$\eatrev$ for the same agent. The following lemma presents such a
relationship when the buyer's demand constraint is a matroid.

\begin{lemma}[\bf Single-agent approximation]
\label{lem:approx-partition}
  Let $\dist$ be any product value distribution and $\feas$ be a matroid with
  feasible polytope $\ptope$. Then, for any $\eaprob\in \frac 12\ptope$, there
  exists a submatroid $\feas' \subseteq \feas$ such that
  \[ \earev(\dist, \feas) \le 33.1\,\eatrev(\dist, \feas') \]
  If $\feas$ is a partition matroid, then $\feas' = \feas$.
\end{lemma}

Putting Lemmas~\ref{lem:relaxation} and \ref{lem:approx-partition},
and Corollary~\ref{cor:stitching} together, and observing that by the
concavity of the revenue objective, $\earev[\frac
12\eaprobi](\disti,\feasi)\ge \frac 12\earev[\eaprobi](\disti,\feasi)$
for all $i$, we get our main result.

\begin{theorem}
\label{thm:main-partition}
  For any product value distribution $\dists$ and feasibility constraints
  $\feas=\{\feasi\}_{i\in[n]}$, where each $\feasi$ is a matroid, there exist
  submatroids $\feasi' \subseteq \feasi$ and a supply-feasible
  $\{\feasi'\}$-limited sequential two-part tariff mechanism $\mech$ such that
  \[ \rev(\dists, \feas)\le 133\,\revm(\dists) \]
  If $\feasi$ is a partition matroid, then $\feasi' = \feasi$.
\end{theorem}

\subsection*{Further Results}


As a consequence of our single-agent approximation
(Lemma~\ref{lem:single-agent} in Section~\ref{sec:single-agent}), we also
obtain an improved approximation for the single-agent revenue maximization
problem with constrained additive values. Specifically, taking $\eaprobs =
\vec{\mathbf{1}}$ and noting $\eaprices = \vec{\mathbf{0}}$,
Lemma~\ref{lem:single-agent} gives the following bound on the optimal revenue
for the single-agent setting.
\begin{corollary}
\label{cor:true-single-agent}
For any downward closed feasibility constraint $\feas$ and any
value distribution $\dists$,
 \[\rev(\dists,\feas) \leq 31.1\,\max\left\{\srev(\dists,\feas),
        \brev(\dists,\feas)\right\}.\]
\end{corollary}

Also as a consequence of Lemma~\ref{lem:single-agent}, we show the following
bound for revenue maximization under an arbitrary \exante\ constraint.
\begin{corollary}
\label{cor:general-ex-ante}
  Let $\dist$ be any product value distribution and $\feas$ be a matroid. Then
  for any $\eaprob\in[0,1]^m$, there exists a submatroid $\feas' \subseteq
  \feas$ such that
  \[\earev(\dist, \feas) \le 35.1\,\eatrev(\dist,\feas') \]
  If $\feas$ is a partition matroid, then $\feas' = \feas$.
\end{corollary}

\section{The \ExAnte\ Relaxation and Stitching}
\label{sec:relaxation}

In this section we prove Lemmas~\ref{lem:relaxation} and
~\ref{lem:stitching-trevs}. 




\begin{proofof}{Lemma~\ref{lem:relaxation}}
 Let $\mech$ be the optimal mechanism for feasibility constraints
 $\feas$ and value distributions $\dists$ achieving revenue
 $\rev(\dists,\feas)$. We will now consider a buyer $i$ and construct a
 mechanism $\mech_i$ for this buyer as follows. When the buyer $i$
 reports a value vector $\vali$, the mechanism $\mech_i$ draws value
 vectors $\tvalsmi$ from the joint distribution $\distsmi$; It
 then returns the allocation and payment that $\mech$ returns at
 $(\vali, \tvalsmi)$. It is easy to see that if $\mech$ is
 BIC, then so is $\mech_i$. Furthermore, $\mech_i$ obtains the same
 revenue from buyer $i$ as $\mech$. Therefore, we have:
 \[ \rev(\dists,\feas) = \sum_i \revm[\mech_i](\disti).\]

 Now let $\alloci$ denote the allocation rule of $\mech_i$ and let
 $\eaprobij = \expect[\vali\sim\disti]{\allocij(\vali)}$. Then,
 recalling equation~\eqref{eq:constrained-revenue}, we have
 $\revm[\mech_i](\disti) \le \earev[\eaprobi](\disti,\feasi)$, and so,
 \[\rev(\dists,\feas)\le \sum_i \earev[\eaprobi](\disti,\feasi).\]

 Finally, the demand feasiblity of $\mech$ implies that the vector
 $\eaprobi$ lies in the polytope $\ptopei$, while the supply
 feasiblity of $\mech$ implies that $\sum_i \eaprobij\le 1$ for all
 $j$. This completes the proof.
\end{proofof}




\begin{proofof}{Lemma~\ref{lem:stitching-trevs}}
  For every buyer $i$, let $\efi$ and $(\priceij[1], \cdots,
  \priceij[m])$ denote the entry fee and item prices respectively in
  the mechanism $\mech_i$. We will compose the mechanisms $\mech_i$ to
  obtain a single mechanism $\mech$ as follows.

  The mechanism $\mech$ considers buyers in an arbitrary order and
  offers items for sale sequentially to the buyers in that order. When
  it is buyer $i$'s turn, some (random set of) items have already been
  sold to other buyers. The mechanism offers the remaining items to
  buyer $i$ via a two-part tariff: it charges the buyer an entry fee
  of $\half\efi$ for the right to buy any subset of the remaining
  items, with item $j$ priced at $\priceij$. Importantly, buyers must
  make the decision of whether or not to participate (that is, whether
  or not to pay the entry fee) before knowing which items are left
  unsold.

  By definition, the mechanism is BIC: buyers may choose whether or
  not to participate and which subset of items to purchase.

  Let us now consider a single buyer $i$. We first claim that when the
  mechanism $\mech$ considers buyer $i$, for every item $j$, the
  probability (taken over value vectors of other agents) that item $j$
  is available to be bought by $i$ is at least $1/2$. Recall that for every pair $i,j$,
  $\prob[\valij\sim\distij]{\valij>\priceij}=1-\distij(\priceij)\le 1-\distij(\eapriceij)=\eaprobij$. So the
  probability that some agent $i'$ buys item $j$ is at most
  $\eaprob_{i'j}$.  Therefore, the probability (over values of agents
  other than $i$) that item $j$ is allocated to an agent other than
  $i$ is at most $\sum_{i'\ne i} \eaprob_{i'j}\le 1/2$, and this
  proves the claim. 

  We will now use the above claim to argue that if after drawing his
  value vector the buyer chooses to participate (i.e. pay the entry
  fee) in mechanism $\mech_i$, then he chooses to participate in
  $\mech$. If agent $i$ participates in mechanism $\mech_i$, then for
  some set $S\in \feas_i$ his value vector satisfies $\sum_{j\in S}
  (\valij-\priceij)-\efi>0$. In the mechanism $\mech$, the agent
  derives from the same set $S$ an expected utility of 
  \[\left(\sum_{j\in S} \prob{j\text{ is
      available for } i}(\valij-\priceij) \right) -\half\efi,\]
  which by the above claim is at least $1/2(\sum_{j\in S}
  (\valij-\priceij)-\efi)>0$. Consequently, if $\mech_i$ obtains the
  entry fee $\efi$ from agent $i$, then $\mech$ obtains the entry fee
  $\efi/2$.

  Next we claim that if agent $i$ buys item $j$ in mechanism $\mech_i$
  and item $j$ is available for him in mechanism $\mech$, then the
  agent buys item $j$ in $\mech$. This follows directly from
  Lemma~\ref{lem:matroid-greedy} (Appendix~\ref{sec:matroids}) by
  noting that $\mech_i$ and $\mech$ offer the same item prices to the
  agent and that the agent is a utility maximizer. As argued
  previously, item $j$ is available with probability at least $1/2$,
  therefore, this claim implies that if $\mech_i$ obtains the price
  $\priceij$ from agent $i$, then $\mech$ obtains the same price
  $\priceij$ with probability $1/2$. Putting this together with the
  above observation about entry fee, we get that $\mech$ obtains in
  expectation at least half of the total revenue obtained by the
  mechanisms $\mech_i$.
\end{proofof}

The proof of the lemma relies upon three facts: (1) mechanism $\mech$
offers each item with probability at least half to each buyer, (2)
under these probabilities, the buyer's expected utility from a set $S$
is at least half his utility from obtaining $S$ with certainty, and,
(3) in the composed mechanism, the buyer selects those items in $S$
that are still available. Fact (2) holds more generally for a buyer
with any monotone submodular value function~\cite{FMV-11}. Fact (3) follows
directly from the definition of gross substitutes valuations,\footnote{A
  valuation $v$ satisfies the gross substitutes condition if for all price
  vectors $\prices$, $\prices'$ where $\prices \leq \prices'$, for all $S$ such
  that $v(S) - \prices \geq v(S') - \prices$ for all $S'$, there exists $T$ such
  that $v(T) - \prices' \geq v(T') - \prices'$ for all $T'$ and $\{j \in S :
  \pricej = \pricej'\} \subseteq T$.} a special case of submodular value
functions. So Lemma~\ref{lem:stitching-trevs} holds more generally for buyers
with gross substitutes valuations.

\section{Two-Part Tariffs for a Single Agent}
\label{sec:single-agent}

We now turn to bounding the revenue from a single agent subject to an \exante\
constraint. In this section we will prove Lemma~\ref{lem:approx-partition}.
In the following discussion, we assume that the buyer has a product
value distribution $\dist=\prod_j\distj$,
and faces a demand feasibility constraint $\feas$, while the mechanism is
subject to an \exante\ supply constraint $\eaprobs$. Recall that we define the
\exante\ prices $\eaprices$ as $\eapricej = \distj^{-1}(1-\eaprobj)$ for all
items $j$.

%
%
\subsubsection*{Core-Tail Decomposition with \ExAnte\ Constraints}

We begin by defining the notation for the core-tail decomposition (see
Table~\ref{tab:notation}). Let
$\tconst \geq 0$ be a constant to be defined later. We use $\threshj =
\eapricej+\tconst$ to denote the threshold for classifying values into
the core or the tail. Specifically, for any item $j$, if $\valj >
\threshj$, we say item $j$ is in the tail, otherwise it is in the
core. Let $\corej$ (resp., $\tailj$) denote the distribution for item
$j$'s value conditioned on the item being in the core (resp., tail).

For a set $A\subseteq [m]$ of items, let $\tprobA$ denote the
probability that the items in $A$ are in the tail and the remaining
items are in the core; that is, $\tprobA=\left(\prod_{j\in
    A}\prob[\valj\sim\distj]{\valj > \threshj}\right)
\left(\prod_{j\not\in A}\prob[\valj\sim\distj]{\valj \le
    \threshj}\right)$. Then $\tprobz$ denotes the probability that all
items are in the core. Observe that as we increase the constant
$\tconst$ (thereby increasing the core-tail thresholds uniformly), the
probability $\tprobz$ increases. We pick $\tconst$ to be the smallest
non-negative number such that $\tprobz\ge 1/2$. Observe that
$\tconst>0$ implies\footnote{For simplicity, we are assuming that the
  value distribution does not contain any point masses; it is easy to
  modify our argument to work in the absence of this assumption, but
  we omit the details.}  $\tprobz=1/2$.

We now state our version of the core-tail decomposition, extended to
respect \exante\ constraints. We defer the proof to
Section~\ref{sec:core-decomp}. Note that although the sum over tail
revenues does not explicitly enforce the \exante\ constraints, the
tail distributions are supported only on values above the \exante\
prices $\eaprices$.

\begin{lemma}[\bf Core Decomposition with \ExAnte\ Constraints]
\label{lem:core-decomposition}
For any product distribution $\dists$, feasibility constraint $\feas$,
and \exante\ constraint $\eaprobs$,
\[\earev(\dists,\feas) \leq \eaVal(\core,\feas) + \sum_{A\subseteq
        [m]}\tprobA\rev(\tail,\feas|_A)\]
\end{lemma}


\begin{table}[t]
  \renewcommand{\arraystretch}{1.5}
  \caption{Notation for Section~\ref{sec:single-agent}.}
  \begin{tabular}{r l l}
    \hline
    Notation & Definition & Formula \\
    \hline
    $\eaprobs$       & \Exante\ probabilities & \\
    $\eaprices$      & \Exante\ prices &
        $\eapricej = \distj^{-1}(1-\eaprobj)\,\,\forall j\in [m]$ \\
    $\threshj$       & Core-tail threshold for item $j$ &
        $\eapricej+\tconst$ \\
    $\tconst$        & Difference between $\threshj$ and $\eapricej$; same for
        all items & $\min\{t \mid  \prob[\vals\sim\dists]{\valj \le
        t+\eapricej \forall j} \ge 1/2 \} $ \\
    $\corej$         & Core distribution for item $j$ &
        $\distj|_{{\valj \leq \threshj}}$ \\
    $\tailj$         & Tail distribution for item $j$ &
        $\distj|_{{\valj > \threshj}}$ \\
    $\coreA$         & Core distribution for items not in $A$ &
        $\prod_{j\not\in A}\corej$ \\
    $\tail$          & Tail distribution for items in $A$ &
        $\prod_{j\in A}\tailj$ \\
    $\tprobj$        & Probability item $j$ is in the tail &
        $\prob[\valj\sim\distj]{\valj > \threshj}$ \\
    $\tprobA$        & Probability exactly items in $A$ are in the tail &
        $\left(\prod_{j\in A}\tprobj\right)
        \left(\prod_{j\not\in A}(1-\tprobj)\right)$ \\
    $\dists-\prices$ & Distribution $\dists$ shifted to the left by $\prices$ &
        \\
    \hline
  \end{tabular}
  \label{tab:notation}
\end{table}

  \subsection{Bounding the Tail}

We first show that the tail revenue can be bounded by selling items
separately under the given \exante\ supply constraint $\eaprobs$. The main
result of this section is as~follows.

\begin{lemma}
\label{lem:tail-bound}
For any product distribution $\dists$ over $m$ independent items and
any $\feas$,
\[\sum_{A\subseteq [m]}\tprobA\rev(\tail,\feas|_A) \leq 8(1+\ln 2) \easrev(\dists,\feas)\]
\end{lemma}

\begin{proof} 
We make use of the following weak but general relationship between the optimal
revenue and the revenue generated by selling separately for a single-agent
constrained additive value setting; this follows by noting that $\rev$ and
$\srev$ are within a factor of $4$ of each other for unit demand agents (see
Appendix~\ref{sec:single-agent-proofs} for a proof). 
\begin{claim}
\label{lem:RevSRevBound}
For any product distribution $\dists$ over $m$ items and any $\feas$,
\[\rev(\dists, \feas) \leq 4m\srev(\dists,\udfeas).\]
\end{claim}
Applying this claim to the revenues $\rev(\tail,\feas|_A)$, we get
that
\[
\sum_{A}\tprobA\rev(\tail,\feas|_A) \leq 4 \sum_A\tprobA|A|\srev(\tail,\udfeas).
\]
We will now use the fact that the tail contains few items in expectation. Let
$\tprobj$ denote the probability that item $j$ is in the tail: $\tprobj =
\prob[\valj\sim\distj]{\valj > \threshj}$. We can write the following series of
inequalities.
\begin{align}
\label{eq:1} \sum_{A}\tprobA|A|\srev(\tail,\udfeas) &\leq 
        \sum_{A}\tprobA|A|\sum_{j\in A}\rev(\tailj) \\
 &\notag = \sum_{j\in[m]}\rev(\tailj)\sum_{A\ni j}\tprobA|A| \\
 &\notag = \sum_{j\in[m]}\tprobj\rev(\tailj)\expect{\lvert A\rvert \, | j \in A} \\
 & \label{eq:4} \leq (1+\ln 2)\sum_{j\in[m]}\earev[\tprobj](\distj) \\
 & \label{eq:5} \leq \frac{1}{\tprobz}(1+\ln 2)\easrev[\tprobs](\dists,\feas)
\end{align}

Here inequality~\eqref{eq:1} follows by removing the demand constraint
$\udfeas$.
Inequality~\eqref{eq:4} follows from three observations: (1) the tail is
non-empty with probability at most $1/2$; (2) if $\{z_i\}_{i\in[n]}$ are
probabilities satisfying $\prod_i(1-z_i)\ge 1/2$, then $\sum_iz_i \leq \ln 2$;
(3) a single-agent single-item mechanism for value distribution $\tailj$ that
achieves revenue $\rev(\tailj)$ would achieve $\tprobj$ times that revenue on
the value distribution $\distj$ while satisfying an \exante\ supply constraint
of $\tprobj$.
Inequality \eqref{eq:5} follows from the standard argument that the revenue
obtained by selling each item individually at prices $\threshj$ (or higher) is
at least $\tprobz$ times the sum of the corresponding per-item revenues.
Finally, the result follows by recalling that $\tprobz\ge 1/2$ and relaxing the
\exante\ constraint.
\end{proof}

  \subsection{Bounding the Core}

Recall that an item $j$ is in the core if its value $\valj$ is no more than the
threshold $\threshj = \eapricej + \tconst$. We will bound the \exante\
constrained social welfare of the core, $\eaVal(\core,\feas)$, in two parts: the
welfare obtained from values below $\eaprices$ via a prophet inequality and the
welfare between $\eaprices$ and $\eaprices + \tconst$ using a concentration
bound introduced by \citet{rw-15}.

Recall that $\corej$ denotes the value distribution for item $j$
conditioned on being in the core. We use $\corej-\eapricej$ to denote
the distribution of $\valj-\beta$ conditioned on $\valj$ being in the
core; in other words, $\corej-\eapricej$ is the distribution $\corej$
shifted to the left by $\eapricej$. $\core-\eaprices$ is defined to be
the product of the distributions $\corej-\eapricej$. Observe that value
vectors drawn from $\core-\eaprices$ are bounded by $\tconst$ in every
coordinate.  The following lemma breaks $\eaVal(\core,\feas)$ up into
the two components, each of which can be bounded separately.


\begin{lemma}
\label{lem:core-bound}
For any product disribution $\dists$ and downwards closed feasibility
constraint $\feas$,
 $\eaVal(\core,\feas) \leq \eaprices\cdot\eaprobs +
 \Val(\core-\eaprices,\feas)$.
\end{lemma}

\begin{proof}
Let $\allocs(\vals)$ be the interim allocation rule of a $\eaprobs$-constrained
BIC mechanism which attains social welfare equal to $\eaVal(\core, \feas)$. Then
\begin{align*}
\eaVal(\core,\feas) &= \sum_j\int_0^{\threshj}\fj(y)\allocj(y)y\dy \\
 &\leq \sum_j\int_0^{\threshj}\fj(y)\allocj(y)\eapricej \dy + \sum_j\int_0^{\threshj}\fj(y)\allocj(y)(y-\eapricej)\dy \\
&\leq \eaprices\cdot\eaprobs + \Val(\core-\eaprices, \feas).
\end{align*}
\end{proof}

We can recover $\Val(\core-\eaprices,\feas)$ using a two-part
tariff for the original distribution $\dists$ by employing the following
concentration result proved by \citet{rw-15}, based on a result of
\citet{schechtman-99}.


\begin{lemma}[\citet{rw-15}]
\label{lem:schechtman}
  Let $\vals$ be a constrained additive value function with a
  downwards closed feasibility constraint, drawn from a distribution over
  support $(-\infty, \tconst]$ for some $\tconst\ge 0$. Let $a$ be the median of
  the value of the grand bundle, $\vals([m])$. Then, $\expect{\vals([m])} \leq
  3a + 4\tconst/\ln 2$.
\end{lemma}

\begin{lemma}
\label{lem:lipschitz-dist-bound}
\begin{align*}
  \Val(\core-\eaprices,\feas) & \leq 6\,\brev(\dists-\eaprices,\feas) +
      \frac{8}{\ln 2} \,\easrev(\dists,\feas) 
\end{align*}
\end{lemma}

\begin{proof}
We apply Lemma~\ref{lem:schechtman} to the distribution $\core-\eaprices$ to
obtain $\Val(\core-\eaprices,\feas) \le 3a+4\tconst/\ln 2$ where $a$ is the
median of the value of the grand bundle under the distribution
$\core-\eaprices$, and $\tconst$ is the constant defined earlier. 




Consider offering the grand bundle at price $a$ to a buyer with value
drawn from $\core-\eaprices$; the buyer accepts with probability
$1/2$. Therefore $\brev(\dists-\eaprices)\ge\brev(\core-\eaprices)\ge
a/2$. Next, suppose that $\tconst>0$. Consider selling the items
separately at prices $\threshj$ for all $j$. Recall that $\tconst>0$
implies that $\prob[\vals\sim\dists]{\exists j \text{ s.t. } \valj >
  \threshj} = 1-\tprobz= 1/2$. So the agent buys at least one item
with probability $1/2$. Noting that $\threshj>\tconst$ for all $j$,
this item pricing obtains a revenue of at least $\tconst/2$. Since
also $\threshj \geq \eapricej$ for all $j$, we have $\tconst \leq
2\easrev(\dists,\feas)$.
\end{proof}

\subsection{Putting the Pieces Together}

Combining Lemmas~\ref{lem:core-decomposition}, \ref{lem:tail-bound},
\ref{lem:core-bound}, and \ref{lem:lipschitz-dist-bound} together, we
obtain the main result of this section:

\begin{lemma}
\label{lem:single-agent}
For any product value distribution $\dists$, downward closed feasibility
constraint $\feas$ and \exante\ constraints $\eaprobs$, 
\[\earev(\dists,\feas) \leq 6\,\brev(\dists-\eaprices,\feas) +
    8(1+\ln 2 + 1/\ln 2)\,\easrev(\dists,\feas) + \eaprices\cdot\eaprobs.\]
\end{lemma}

It remains to bound the $\eaprices\cdot\eaprobs$ term. Note that this term is
the revenue that would be obtained in the absence of any demand constraint
(equivalently, in the additive setting) by setting the \exante\ prices on the
items. When $\feas$ is a partition matroid and if the \exante\
constraint $\eaprobs$ lies in the shrunk polytope $\frac 12\ptope$,
\citet{CHMS-STOC10} show via a prophet inequality that the term
$\eaprices\cdot\eaprobs$ is bounded by the revenue of an item pricing.

\begin{lemma}[\citet{CHMS-STOC10}]
\label{thm:partition-matroid}
For a partition matroid $\feas$, \exante\ constraints
$\eaprobs\in\frac 12\ptope$, and corresponding \exante\ prices
$\eaprices$,
 \[\eaprices\cdot\eaprobs \leq 2\,\easrev(\dists,\feas).\]
\end{lemma}

No prophet inequality based on static thresholds is known for general
matroids. However, \citet{fsz-15} nonetheless show that, if $\eaprobs
\in b\ptope$, selling at the \exante\ prices recovers a $(1-b)$
fraction of the relaxed revenue under a stronger demand constraint.
This leads to the following result.

\begin{lemma}[\citet{fsz-15}]
\label{thm:ocrs}
  For a general matroid $\feas$, constant $b \in (0,1)$, \exante\ constraints
  $\eaprobs \in b\ptope$, and corresponding \exante\ prices $\eaprices$, there
  exists a submatroid $\feas' \subseteq \feas$ such that
  \[\eaprices\cdot\eaprobs \leq \frac{1}{1-b}\easrev(\dists,\feas').\]
  Furthermore, the constraint $\feas'$ is efficiently computable.
\end{lemma}

We are now ready to prove Lemma~\ref{lem:approx-partition} and
Corollary~\ref{cor:general-ex-ante}, stated in
Section~\ref{sec:theorems}.

\begingroup
\def\thetheorem{\ref{lem:approx-partition}}
\begin{lemma}
  Let $\dist$ be any product value distribution and $\feas$ be a matroid with
 feasible polytope $\ptope$. Then, for any $\eaprob\in \frac 12\ptope$, there
 exists a submatroid $\feas' \subseteq \feas$ such that
 \[ \earev(\dist, \feas) \le 33.1\,\eatrev(\dist, \feas') \]
 If $\feas$ is a partition matroid, then $\feas' = \feas$.
\end{lemma}
\addtocounter{theorem}{-1}
\endgroup

\begin{proof} 
  We first observe that $\brev(\dists-\eaprices,\feas) \le
  \eatrev(\dists,\feas)$. In particular, for any $a>0$, a two-part
  tariff with entry fee $a$ and item prices $\eaprices$ achieves at
  least as much revenue over values drawn from $\dists$ as does a
  bundle pricing with price $a$ over values drawn from $\dists -
  \beta$. The lemma now follows from Lemma~\ref{lem:single-agent},
  together with the bounds on $\eaprices\cdot\eaprobs$ given by
  Lemmas~\ref{thm:partition-matroid} and \ref{thm:ocrs}
\end{proof}


\smallskip As a final remark, we note that the condition $\eaprobs \in
\frac12\ptope$ in Lemma~\ref{lem:approx-partition} is necessary only
to recover $\eaprices\cdot\eaprobs$; we can, in fact, show a slightly
weaker result which holds for arbitrary $\eaprobs$. 

\begingroup
\def\thetheorem{\ref{cor:general-ex-ante}}
\begin{corollary}
  Let $\dist$ be any product value distribution and $\feas$ be a matroid. Then
  for any $\eaprob\in[0,1]^m$, there exists a submatroid $\feas' \subseteq
  \feas$ such that
  \[\earev(\dist, \feas) \le 35.1\,\eatrev(\dist,\feas') \]
  If $\feas$ is a partition matroid, then $\feas' = \feas$.
\end{corollary}
\addtocounter{theorem}{-1}
\endgroup

\begin{proof}
For any $\eaprobs
\in [0,1]^m$,
\[\earev(\dist,\feas) \leq
    \max_{\substack{\eaprobs'\leq\eaprobs \\ \eaprobs'\in\ptope}}
    \earev[\eaprobs'](\dist,\feas).\]
Therefore, there exists $\eaprobs' \in \ptope$ and corresponding
$\eaprices'$, such that Lemma~\ref{lem:single-agent} gives
\[\earev(\dists,\feas) \leq 31.1 \eatrev[\eaprobs'](\dist,\feas) +
    \eaprices'\cdot\eaprobs'.\]
Furthermore, by scaling $\eaprobs'$ to lie in $\frac12\ptope$ we can only
increase the corresponding \exante\ prices, so Lemma~\ref{thm:ocrs} gives
$\eaprices'\cdot\eaprobs' \leq 4\easrev[\eaprobs'](\dist,\feas')$ for some
$\feas'\subseteq\feas$.  The corollary now follows by noting
$\eatrev[\eaprobs'](\dist,\feas) \leq \eatrev(\dist,\feas)$ and
$\easrev[\eaprobs'](\dist,\feas') \leq \easrev(\dist,\feas')$. 
\end{proof}

  \subsection{Core Decomposition with \ExAnte\ Constraints}
\label{sec:core-decomp}

The proof of Lemma~\ref{lem:core-decomposition} makes use of the following two
lemmas, which are analogous to results proved by Babaioff et al. after Li and
Yao. 

\begin{lemma}
\label{lem:subdomain-stitching}
 There exists a set $\{\eaprobsA \in [0,1]^m : A \subseteq [m]\}$ such that
 $\sum_A \tprobA \eaprobAj \leq \eaprobj$ for all $j$ and
 \[\earev(\dists,\feas) \leq
 \sum_{A\subseteq[m]}\tprobA\earev[\eaprobsA](\dists_A,\feas)\]
\end{lemma}

\begin{proof}
Let $\mech$ be a BIC mechanism which is $\eaprobs$-constrained under $\dists$
such that $\revm(\dists) = \earev(\dists,\feas)$. So $\revm(\dists) =
\sum_{A\subseteq[m]}\tprobA\revm(\dists_A)$. Let $\eaprobAj$ be the probability
that $\mech$ allocates item $j$ given that $\vals$ is drawn from $\dists_A$;
that is, $\eaprobAj = \expect[\vals\sim\dists_A]{\allocj(\vals)}$.  Clearly
$\sum_{A}\tprobA \eaprobAj \leq \eaprobj$, by the assumption that $\mech$ is
$\eaprobs$-constrained. The result follows since $\revm(\dists_A) \leq
\earev[\eaprobsA](\dists_A,\feas)$ for each $A$.
\end{proof}

\begin{lemma}
\label{lem:marginal-mechanism}
 For any two independent distributions $\dists_S$ and $\dists_T$ over disjoint
 sets of items $S$ and $T$ with corresponding \exante\ constraints
 $\eaprobs_S$ and $\eaprobs_T$ and a joint feasibility constraint $\feas$,
 \[\earev[(\eaprobs_S; \eaprobs_T)](\dists_S\times\dists_T,\feas)
        \leq \eaVal[\eaprobs_S](\dists_S,\feas|_S) +
             \earev[\eaprobs_T](\dists_T,\feas|_T).\]
\end{lemma}

\begin{proof}
Let $\mech$ be a BIC mechanism which is $(\eaprobs_S; \eaprobs_T)$-constrained
under $(\dists_S\times\dists_T)$ such that $\revm(\dists_S\times\dists_T) =
\earev[(\eaprobs_S;\eaprobs_T)](\dists_S\times\dists_T,\feas)$.  We construct a
mechanism $\mech'$ for selling items in $T$ as follows. $\mech'$ first samples
$\vals_S\sim\dists_S$, and then solicits a bid $\vals_T$ for items in $T$. Let
$(\allocs_{S\union T}(\vals_S;\vals_T), \price(\vals_S;\vals_T))$ be the
allocation returned and payment charged by $\mech$ for the combined bid; then
$\mech'$ returns the allocation $\allocs_T(\vals_S;\vals_T) $ and charges
$\price(\vals_S;\vals_T) - \vals_S(\allocs_S(\vals_S;\vals_T))$.

We now prove that $\mech'$ is truthful. Suppose the bidder submits a bid
$\vals_T'$. His utility is $\vals_T(\allocs_T(\vals_S; \vals_T')) -
\big(\price(\vals_S; \vals_T') - \vals_S(\allocs_S(\vals_S; \vals_T'))\big)$,
which is the utility of a bidder participating in $\mech$ with valuation
$(\vals_S,\vals_T')$. Since $\mech$ is truthful, the bidder can do no worse by
bidding $\vals_T$ in $\mech'$ and receiving the utility of an agent who bids
truthfully in $\mech$.

Note that $\mech'$ allocates item $j \in T$ exactly when $\mech$ does
(conditioned on $\vals_S$). So $\mech'$ is demand-feasible. Furthermore, since
$\mech'$ draws $\vals_S$ from $\dists_S$, $\mech'$ is also
$\eaprobs_T$-constrained under $\dists_T$.  Formally, let $\allocs'$ be the
allocation rule of $\mech'$; then
$\expect[\vals_T\sim\dists_T]{\allocj'(\vals_T)} =
\expect[\vals_S\sim\dists_S,\vals_T\sim\dists_T]{\allocj(\vals_S;\vals_T)} \leq
\eaprobj$ for all $j \in T$.

The revenue obtained by $\mech'$ is
\begin{align*}
\revm[\mech'](\dists_T) &=
        \expect[\vals_S\sim\dists_S,\vals_T\sim\dists_T]{\price(\vals_S;\vals_T)
        - \vals_S(\allocs_S(\vals_S;\vals_T))} \\
 &= \earev[(\eaprobs_S;\eaprobs_T)](\dists_S\times\dists_T,\feas)
        - \expect[\vals_S\sim\dists_S,\vals_T\sim\dists_T]{\vals_S(
        \allocs_S(\vals_S;\vals_T))} \\
 &\geq \earev[(\eaprobs_S;\eaprobs_T)](\dists_S\times\dists_T,\feas) -
        \eaVal[\eaprobs_S](\dists_S,\feas|_S),
\end{align*}
where the inequality follows because the welfare $\mech$ obtains from items
in $S$ is a lower bound on the welfare of any $\eaprobs_S$-constrained mechanism
for $\dists_S$.
\end{proof}

\begin{proofof}{Lemma \ref{lem:core-decomposition}}
By Lemmas \ref{lem:subdomain-stitching} and \ref{lem:marginal-mechanism}, we
have
\[\earev(\dists,\feas) \leq \sum_{A\subseteq[m]} \tprobA\left(
        \eaVal[\eaprobsA](\coreA,\feas|_{A^c}) +
        \earev[\eaprobsA](\tail,\feas|_{A})\right).\]

For each $A\subseteq [m]$, let $\mech^A$ be a truthful $\eaprobsA$-constrained
demand-feasible mechanism which obtains welfare equal to
$\eaVal[\eaprobsA](\coreA,\feas|_{A^c})$.  One way to allocate items when values
are drawn from $\core$ is to choose to sell only items from some set
$A\subseteq[m]$. Consider a mechanism which chooses from among all subsets of
items, choosing $A$ with probability $\tprobA$, and then runs $\mech^A$.  The
expected welfare from such a mechanism is exactly
$\sum_{A\subseteq[m]}\tprobA\eaVal[\eaprobsA](\coreA,\feas|_{A^c})$.  Since
$\sum_{A\subseteq[m]}\tprobA\eaprobAj \leq \eaprobj$, the welfare of this
mechanism also provides a lower bound on $\eaVal(\core,\feas)$.
\end{proofof}

\section{Approximation for Symmetric Agents}
\label{sec:symmetric}

Computing the approximate mechanisms of
Theorem~\ref{thm:main-partition} requires being able to efficiently
solve the \exante\ optimization, $\max_{\eaprobs} \sum_i
\earev[\eaprobi](\disti,\feasi) \; \text{s.t. } \sum_i
\eaprobi\le\vec{\mathbf 1}$. This is not necessarily a convex
optimization problem and it is not clear whether this can be solved or
approximated efficiently in general. In this section we show how to
solve this problem in the special case where agents are a priori
identical.


In a {\em symmetric agents} setting, agents share a common feasibility
constraint and value distribution. In particular, $\feasi =
\feasi[i']=\feas$ and $\disti = \disti[i']=\dist$ for all $i, i'\in
[n]$.  Note that the values of different items are not necessarily
distributed identically, neither is $\feas$ necessarily symmetric
across items. Since each agent is identical, we can focus on
maximizing the revenue obtained from a single agent, while ensuring
that the \exante\ probability of selling each item is small
enough that we may apply Lemma~\ref{lem:stitching-trevs}. We formalize
this in the following lemma. See Appendix~\ref{sec:single-agent-proofs} for a
proof.
    
\begin{lemma}
\label{lem:symmetric-reduction}
 In a symmetric agents setting with $n$ agents, a matroid feasibility constraint $\feas$ and
 product distribution $\dists$,
 \[\rev(\times_n \dists, \times_n \feas) \leq
    2n \max_{\eaprobs\in \ptope \intersect\left[0,\tfrac{1}{2n}\right]^m}
    \earev[\eaprobs](\dist,\feas),\]
\end{lemma}

For the remainder of this section, we focus on efficiently
approximately maximizing the single agent objective
$\earev[\eaprobs](\dist,\feas)$ subject to $\eaprobs\in
\widehat{\ptope}$, where we use $\widehat{\ptope}$ as short form for
$\ptope
\intersect\left[0,\tfrac{1}{2n}\right]^m$. Lemma~\ref{lem:single-agent}
bounds the revenue by three terms; we observe that $\easrev(\dists,
\feas)$ is at most $\max_{\eaprobs'\le\eaprobs}
\eaprobs'\cdot\dists^{-1}(1-\eaprobs')$. Therefore,
\begin{align}
\notag \max_{\eaprobs\in \widehat{\ptope}}\earev(\dists,\feas) & \leq
6\,\max_{\eaprobs\in \widehat{\ptope}}\brev(\dists-\dists^{-1}(1-\eaprobs),\feas) + 26.1\, \max_{\eaprobs\in \widehat{\ptope}}\eaprobs\cdot\dists^{-1}(1-\eaprobs)\\
&\label{eq:symmetric} \le 6\,\brev(\dists-\dists^{-1}(1-1/2n),\feas) + 26.1\, \max_{\eaprobs\in \widehat{\ptope}}\eaprobs\cdot\dists^{-1}(1-\eaprobs)
\end{align}




\noindent
The first term on the LHS of Equation~\eqref{eq:symmetric} is easy to capture.
We can use sampling to efficiently compute the optimal bundle price
for the value distribution $\dists-\dists^{-1}(1-1/2n)$. Call this
price $a$, and let $\pricej=\distj^{-1}(1-1/2n)$ for all items $j\in
[m]$. Then, by Lemma~\ref{lem:stitching-trevs} the multi-agent
sequential two-part tariff mechanism that offers each agent an entry
fee of $a$ and per item pricing $\prices$ obtains revenue at least
$\frac n2 \brev(\dists-\dists^{-1}(1-1/2n),\feas)$.

This leaves us with the following maximization problem:
\begin{equation}
\label{eq:bq-max}
\begin{aligned}
 \text{maximize }\; & \eaprobs\cdot\dists^{-1}(1-\eaprobs) & \text{s.t. }\; & \eaprobs \in \ptope \intersect \left[0,\tfrac{1}{2n}\right]^m
\end{aligned}
\end{equation}

In Appendix~\ref{sec:optimizing-beta-dot-q} we discuss how to solve (a
relaxation of) this problem efficiently when $\feas$ is a matroid. We
obtain a (potentially random) vector $\eaprobs$ that in expectation
satisfies the feasibility constraint $\widehat{\ptope}$ and obtains an
expected objective function value no smaller than the optimum of
\eqref{eq:bq-max}. 
Then, for partition matroids, we can employ a constructive version of
Lemma~\ref{thm:partition-matroid} due to \citet{CHMS-STOC10} to obtain
a (potentially random) sequential two-part tariff mechanism that
obtains revenue at least $\frac n4$ times the optimum of
\eqref{eq:bq-max}. For general matroids, we can likewise employ a
constructive version of Theorem~\ref{thm:ocrs} due to \citet{fsz-15}
to obtain a (potentially random) demand-limiting sequential two-part
tariff mechanism that obtains revenue at least $\frac n4$ times the
optimum of \eqref{eq:bq-max}. We obtain the following theorem.

\begin{theorem}
  For any symmetric, matroid feasibility constraint $\feas$ and
  symmetric, product distribution $\dists$, there is an efficiently
  computable randomized demand-limiting sequential
  two-part tariff mechanism $\mech$ and a constant $c$ such that
 \[\rev(\dists,\feas) \leq c\, \revm(\dists).\]
 When $\feas$ is a partition matroid, we obtain a sequential two-part
 tariff mechanism, and when $\distj$ is regular for all $j$, our
 mechanism is deterministic. 
\end{theorem}


\section*{Acknowledgements}
We are grateful to Anna Karlin for feedback on early drafts of this
work, and to Jason Hartline for insights on efficiently solving the
ex ante relaxation for symmetric agents.


\bibliographystyle{apalike}
\bibliography{agt,bmd,ea}


\appendix
\section{Matroid Concepts}
\label{sec:matroids}

A matroid $M$ is a tuple $(G, \I)$ where $G$ is called the {\em ground set} and
$\I \subseteq 2^G$ is a collection of {\em independent sets} satisfying the
following two properties:
\begin{enumerate}
  \item If $I \subseteq J$ and $J \in \I$, then $I \in \I$ ($\I$ is
  downward-closed); and

  \item If $I,J \in \I$ and $|J| > |I|$, then there exists $e \in J\setminus I$
  such that $(I\union\{e\}) \in \I$.
\end{enumerate}
A {\em basis} is an independent set of maximal size: $B \subseteq G$
is a basis if $B \in \I$ and $|I| \leq |B|$ for all $I \in \I$. The
following lemma is a simple consequence of the fact that the greedy
algorithm finds the maximum weight basis in any matroid.
  \begin{lemma}
    \label{lem:matroid-greedy}  
    Let $\feas$ be any matroid over ground set $G$, $I$ be any subset of $G$,
    and $w$ be any vector of weights defined on elements in $G$.  If $j\in G$
    belongs to a maximum weight basis of $\feas$ and $j\in I$, then $j$ also
    belongs to a maximum weight basis of $\feas|_{I}$.
  \end{lemma}

Several classes of matroids are of special interest. A {\em $k$-uniform} matroid
is a matroid in which any $S \subseteq G$ with $|S| \leq k$ is an independent
set; the class of uniform matroids generalizes the extensively studied additive
($k=m)$ and unit-demand ($k=1$) settings.  A {\em partition} matroid is the
union of uniform matroids: $G = G_1\union\ldots\union G_N$, where $(G_i,\I_i)$
is a $k_i$-uniform matroid, and a set $S\subseteq G$ is independent if
$S\intersect G_i \in \I_i$ for all $i$.

\section{Omitted Proofs}
\label{sec:single-agent-proofs}

\subsection{Proofs from Section \ref{sec:single-agent}}

We make use of the following result of \citet{cms-10}.
\begin{lemma}
\label{lem:unit-srev}
 (\cite{cms-10}) 
 For any product distribution $\dists$,
 \[\rev(\dists,\udfeas) \leq 4\srev(\dists,\udfeas).\]
\end{lemma}

\begin{proofof}{Claim~\ref{lem:RevSRevBound}}
Let $\mech$ be a BIC mechanism such that $\revm(\dists) =
\rev(\dists,\feas)$. Let $(\allocs(\vals), \price(\vals))$, where
$\sum_{j=1}^m\allocj(\vals) \leq m$, be the lottery offered by $\mech$ to an
agent who reports $\vals$. We modify $\mech$ to get $\mech'$, a BIC
mechanism which allocates at most one item and has revenue
$\revm[\mech'](\dists) = \frac{1}{m}\revm(\dists)$.

For every type $\vals$, let $\x'(\vals) = \frac{1}{m}\allocs(\vals)$ and
$\price'(\vals) = \frac{1}{m}p(\vals)$ be the lottery offered by $\mech'$.
Since $|\allocs(\vals)|_1 \leq m$, we have $|\allocs'(\vals)|_1 \leq 1$, and so
$\mech'$ is feasible for the unit-demand setting.  Because the buyer's utility
is quasi-linear, scaling the allocation probabilities and payments by $m$ simply
scales the utility of each outcome by $m$. Therefore, the buyer will select
corresponding outcomes in $\mech$ and $\mech'$, and $\mech'$ is BIC with
revenue $\frac{1}{m}\revm(\dists)$.  Combined with
Lemma~\ref{lem:unit-srev}, this completes the proof.
\end{proofof}

\subsection{Proofs from Section~\ref{sec:symmetric}}

\begin{proofof}{Lemma~\ref{lem:symmetric-reduction}}
  Let $\mech$ be a demand- and supply-feasible BIC mechanism such that
  $\revm(\dists) = \rev(\dists,\feas)$.  Let $\eaprobij$ be the
  probability with which $\mech$ sells item $j$ to agent $i$. By
  symmetry, we can permute the identities of the agents uniformly at
  random before running $\mech$ without hurting the expected
  revenue. Under this permutation, the ex ante probability with which
  $\mech$ sells $j$ to $i$ is at most $1/n$. We can therefore assume
  without loss of generality that $\eaprobij \leq 1/n$. Now, consider
  a single agent $i$; with probability $1/2$, allocate the empty set to
  this agent at price $0$, and with probability $1/2$, draw values for
  all other agents from $\dists_{-i}$ and simulate mechanism
  $\mech$. The resulting mechanism is a single agent mechanism that
  obtains a $1/2n$ fraction of the revenue of $\mech$ and satisfies an
  ex ante constraint $\eaprobs\in
  \ptope\intersect\left[0,\tfrac{1}{2n}\right]^m$. The lemma follows.
\end{proofof}

\section{Alternate Approximation for Many Additive Agents}
\label{sec:additive}

For the special case of additive buyers, we show how to modify the analysis of
\citet{bilw-focs14} in order to achieve a much tighter bound than that stated in
Theorem~\ref{thm:main-partition}. The relaxation and stitching steps hold as
before; we prove that \citeauthor{bilw-focs14}'s single-agent approximation can
be made to respect \exante\ constraints with only a small penalty to the
approximation factor.

\begin{lemma}
\label{lem:approx-additive}
For any product distribution $\dist$ and any $\eaprobs \in [0,1]^m$,
\[\earev(\dist,\adfeas) \leq 7\,\eatrev(\dists,\adfeas)\]
\end{lemma}

Combining Lemma~\ref{lem:approx-additive} with Lemma~\ref{lem:relaxation} and
Corollary~\ref{cor:stitching} as before, we get our improved result.

\begin{theorem}
\label{thm:main-additive}
  For any product value distribution $\dists$, there exists a supply-feasible
  sequential two-part tariff mechanism $\mech$ such that
  \[ \rev(\dists, \times_n \adfeas)\le 28\,\revm(\dists) \]
\end{theorem}
We devote the remainder of this section to the proof of
Lemma~\ref{lem:approx-additive}; we drop the explicit dependence on the
feasibility constraint, $\adfeas$, for clarity.

We make use of the fact that relaxing the demand constraint is unnecessary
in the additive setting. This allows us to get a much tighter concentration
result in the core.  Instead of defining $\threshj = \eapricej + \tconst$, as in
Section~\ref{sec:single-agent}, we define $\threshj = \max\left(\eapricej,
\easrev(\dist)\right)$. It is straightforward to verify that our core
decomposition (Lemma~\ref{lem:core-decomposition}) continues to hold under this
definition.  The key insight in~\citet{bilw-focs14}'s analysis is that this
definition allows for a nontrivial bound on the variance of the core, leading to
a strong concentration result via Chebyshev's inequality, while keeping the
expected number of items in the tail small. It turns out that their analysis
goes through under an \exante\ constraint, except for a small loss in the core
due to enforcing the constraint for the bundle pricing.


In addition to the notation from Section~\ref{sec:single-agent}, we define the
following notation (see Table~\ref{tab:add-notation}). Let $r_j =
\earev[\eaprobj](\distj)$ and $r = \easrev(\dist)$. Note that in the additive setting
$\easrev(\dist) = \sum_j\earev[\eaprobj](\distj)$; in other words, $r = \sum_jr_j$.

\begin{table}[t]
  \renewcommand{\arraystretch}{1.5}
  \caption{Notation for Section~\ref{sec:additive}.}
  \begin{tabular}{r l l}
    \hline
    Notation & Definition & Formula \\
    \hline
    $\eaprobs$       & \Exante\ probabilities & \\
    $\eaprices$      & \Exante\ prices &
        $\eapricej = \distj^{-1}(1-\eaprobj)\,\,\forall j\in [m]$ \\
    $r_j$            & Revenue from item $j$ &
        $\earev[\eaprobj](\distj)$ \\
    $r$              & Item-pricing revenue &
        $\sum_jr_j$ \\
    $\threshj$       & Core-tail threshold for item $j$ &
        $\max(\eapricej, r)$ \\
    $\tprobj$        & Probability item $j$ is in the tail &
        $\prob[\valj\sim\distj]{\valj > \threshj}$ \\
    $\dists-\prices$ & Distribution $\dists$ shifted to the left by $\prices$ &
        \\
    \hline
  \end{tabular}
  \label{tab:add-notation}
\end{table}


By Lemmas~\ref{lem:core-decomposition} and \ref{lem:core-bound}, we have
\[\earev(\dist) \leq \eaprices\cdot\eaprobs + \Val(\core-\eaprices) +
    \sum_{A\subseteq[m]}\tprobA\rev(\tail)\]

Clearly $\eaprices\cdot\eaprobs \leq \easrev(\dist)$, so it remains to bound the
other  two terms. Note that the \exante\ constraint has been effectively removed
from these terms; we will show that \citeauthor{bilw-focs14}'s unconstrained
bounds continue to apply here. We state these bounds and then show that our
distributions satisfy their conditions.  Recalling that $\brev(\core-\eaprices)
leq \eatrev(\dist)$ completes the proof.

\begin{lemma}[\citet{bilw-focs14}]
\label{lem:add-core}
  If, for all $j$, $\var(\corej-\eapricej) \leq 2rr_j$, then
  \[\Val(\core - \eaprices) \leq
        4\,\max\left\{\brev(\core-\eaprices),\easrev(\dists)\right\}\]
\end{lemma}

\begin{lemma}[\citet{bilw-focs14}]
\label{lem:add-tail}
  If, for all $j$, $\rev(\tailj) \leq r_j/\tprobj$ and $\tprobj \leq r_j/r$,
  then
  \[\sum_{A\subseteq[m]}\tprobA\rev(\tail) \leq 2\,\easrev(\dists)\]
\end{lemma}

The following two lemmas capture the necessary conditions.

\begin{lemma}
  $\var(\corej-\eapricej) \leq 2rr_j$
\end{lemma}

\begin{proof}
  We first prove $\rev(\corej-\eapricej) \leq r_j$ for all $j$.  Let
  $\pricej^* \geq 0$ be an optimal price for selling to
  $\corej-\eapricej$.  Selling to $\corej$ at price $\pricej^* +
  \eapricej$ gets at least as much revenue and sells with probability
  at most $\eaprobj$, so $\rev(\corej-\eapricej) \leq
  \earev[\eaprobj](\corej)$. Now, let $q_j^*$ be the probability with
  which a mechanism obtaining $\earev[\eaprobj](\corej)$ sells to
  $\corej$. Clearly, since $\distj$ stochastically dominates $\corej$,
  selling to $\distj$ with probability $q_j^*$ gets at least as much
  revenue and satisfies the same \exante\ constraint.

Given the above, we employ an argument originally due to \citet{LY-PNAS13} to
bound the variance of $\corej-\eaprobj$.  Note that $\corej-\eapricej$ is
supported on $[0,r]$, but its revenue is at most $r_j$. So
$\prob[V\sim(\corej-\eapricej)]{V \geq v} \leq r_j/v$. Then
\begin{align*}
\expect[V\sim(\corej-\eapricej)]{V^2} &\leq \int_0^{r^2}\min(1,
        r_j/\sqrt{v})\,dv \\
  &\leq 2rr_j
\end{align*}
\end{proof}

\begin{lemma}
$\rev(\tailj) \leq r_j/\tprobj$ and $\tprobj \leq r_j/r$.
\end{lemma}

\begin{proof}
The first inequality follows from the assumption that $\threshj \geq \eapricej$.
Let $\pricej^*$ be an optimal price for $\rev(\tailj)$. Then, by setting price
$\pricej^*$, one can obtain $\tprobj\rev(\tailj)$ from $\distj$ while respecting
the \exante\ constraint $\eaprobj$. In other words,
$\tprobj\rev(\tailj) \le \earev[\eaprobj](\distj) = r_j$.

Recall that $\threshj \ge r$.  So one could sell item $j$ at price
$\threshj$ and earn profit $\tprobj\threshj\ge \tprobj r$ while
respecting the \exante\ constraint $\eaprobj$. But
$\earev[\eaprobj](\distj) = r_j$, therefore we must have $\tprobj r \leq
\tprobj\threshj\leq r_j$.
\end{proof}

\section{Efficient Approximation for Symmetric Agents}
\label{sec:optimizing-beta-dot-q}

We will now discuss how to solve the optimization problem \eqref{eq:bq-max} efficiently when $\feas$ is a matroid. We first modify the distribution $\dists$ so that for every item $j$, any value below quantile $1-1/2n$ is mapped to $0$. The problem then simplifies to the following.
\begin{equation}
\label{eq:bq-max-2}
\begin{aligned}
\text{maximize }\; & \eaprobs\cdot\dists^{-1}(1-\eaprobs) & \text{s.t. }\; & \eaprobs \in \ptope
\end{aligned}
\end{equation}
This problem is related to the \exante\ relaxation of the
single-parameter revenue maximization problem with $m$ buyers, where
buyer $j$'s value is distributed independently according to $\distj$
and the seller faces the feasibility constraint $\feas$ (i.e., he can
sell to any subset of buyers that form an independent set in
$\feas$). When the distributions $\distj$ are all regular, the
objective $\eaprobs\cdot\dists^{-1}(1-\eaprobs)$ is concave, and the
above problem can be solved using standard convex optimization
techniques.

When the distributions $\distj$ are not all regular,
\eqref{eq:bq-max-2} is not necessarily convex. In this case, allowing
for a randomized solution convexifies the problem. 
Consider the following
relaxation that maximizes the objective over all distributions over
vectors $\eaprobs$:
\begin{equation}
\label{eq:bq-max-3}
\begin{aligned}
\text{maximize }\; & \expect{\eaprobs\cdot\dists^{-1}(1-\eaprobs)} & \text{s.t. }\; & \expect{\eaprobs} \in \ptope
\end{aligned}
\end{equation}
This problem can in turn be restated as follows: maximize the ironed
virtual surplus of a BIC mechanism for the single-parameter revenue
maximization problem stated above subject to the feasibility
constraint $\ptope$ imposed \exante. 

\citet{Hartline-comm} describes an alternative to standard convex
optimization for solving the above problem to within arbitrary
accuracy. Pick a sufficiently small $\epsilon>0$. Discretize the
problem by creating a new discrete distribution $\distj'$ for every
$j\in [m]$ as follows: for every integer $z$ in $[0, 1/\epsilon]$,
place a mass of $\epsilon$ on the ironed virtual value for
distribution $\distj$ at quantile $z\epsilon$. Let $R_j$ denote the
support of distribution $\distj'$. Over these discrete supports, the
ironed virtual value maximization problem becomes one of selecting a
subset of $\cup_j R_j$ of maximum total (ironed virtual) value subject
to the constraint that the subset can be partitioned into at most
$1/\epsilon$ parts each of which is independent in $\feas$. In other
words, this is the problem of finding a maximum weight basis over a
matroid formed by the union of $1/\epsilon$ identical copies of
$\feas$. The standard greedy algorithm for matroids solves this
problem efficiently for any constant $\epsilon$. This algorithm
approximates \eqref{eq:bq-max-3} to within an additive error of
$\epsilon\sum_j\distj^{-1}(1)$, and in the case of non-regular
distributions, produces a distribution over two vectors $\eaprobs_1$
and $\eaprobs_2$ that in expectation satisfies the constraint
$\ptope$.

\end{document}